\documentclass[reqno,centertags,12pt]{amsart}
\usepackage{amsmath,amsthm,amscd,amssymb}
\usepackage{latexsym}
\usepackage{graphicx}
\usepackage{enumerate}
\usepackage{ulem} %to get strikethrough in math mode
\usepackage[mathscr]{eucal}

\usepackage{hyperref}
\newcommand{\bbN}{{\mathbb{N}}}

\newcommand{\bbR}{{\mathbb{R}}}
\newcommand{\bbS}{{\mathbb{S}}}

\newcommand{\bbZ}{{\mathbb{Z}}}

\newcommand{\calC}{{\mathcal{C}}}

\newcommand{\calF}{{\mathcal F}}

\newcommand{\bdone}{{\boldsymbol{1}}}

\newcommand{\lb}{\label}

\newcommand{\tr}{\text{\rm{Tr}}}
\newcommand{\ntr}{\text{\sout{\rm{Tr}}}}

\newcommand{\supp}{\text{\rm{supp}}}

\newcommand{\bi}{\bibitem}

\newcommand{\beq}{\begin{equation}}
\newcommand{\eeq}{\end{equation}}
\newcommand{\ba}{\begin{align}}
\newcommand{\ea}{\end{align}}

\newcommand{\jap}[1]{\langle #1 \rangle}

\newcommand{\Norm}[1]{\lVert#1\rVert}

\newcommand{\sgn}{\mathrm{sgn}}

\newcounter{smalllist}

\newtheorem{theorem}{Theorem}[section]
\newtheorem{proposition}[theorem]{Proposition}
\newtheorem{lemma}[theorem]{Lemma}
\newtheorem{corollary}[theorem]{Corollary}

\theoremstyle{definition}
\newtheorem{example}[theorem]{Example}

\newtheorem*{remark}{Remark}
\newtheorem*{remarks}{Remarks}

\allowdisplaybreaks
\numberwithin{equation}{section} % if not working place after \maketitle

\begin{document}

\title{Comparison of Ising Models Under Change of Apriori Measure}
\author[J.~Madrid, B.~Simon and D.~R.~Wells]{Jos\'{e} Madrid$^{1}$, Barry Simon$^{2,3}$, Daniel R. Wells}

\thanks{$^1$ Department of Mathematics, University of California Los Angeles, Portola Plaza 520, Los Angeles, CA 90095, USA. E-mail: jmadrid@math.ucla.edu}

\thanks{$^2$ Departments of Mathematics and Physics, Mathematics 253-37, California Institute of Technology, Pasadena, CA 91125. E-mail: bsimon@caltech.edu}

\thanks{$^3$ Research supported in part by Israeli BSF Grant No.~2020027.}

\dedicatory{Dedicated with pleasure to Elliott Lieb on his $90^{th}$ birthday.}

\date{\today}
\keywords{Ising Model, Correlation Functions, $D$-vector Model, Spin $S$ Models, Mean Field Theory, Majorization}
\subjclass[2020]{82B20, 26D20, 46N10}

\begin{abstract}
We study comparison of correlation functions for ferromagnetic generalized Ising models with two different apriori measures.  One purpose of this note is to publicize some unpublished 45 year old work of Daniel Wells on the issue.  We then prove results for the apriori measures associated to one component of $D$-vectors uniformly distributed on the unit sphere and also the case of spin $S$ ($2S+1$ equally spaced values symmetric about $0$ and with equal weights) that improves some 50 year old bounds of Griffiths on transition temperatures.
\end{abstract}

\maketitle

%%%%%%%%%%%%%%%%%%%%%%%%%%%%%%%%%%%%%%%%%%%%%%%%%%%%%%%%%%%%%%
\section{Introduction} \lb{s1}
%%%%%%%%%%%%%%%%%%%%%%%%%%%%%%%%%%%%%%%%%%%%%%%%%%%%%%%%%%%%%%

Besides Elliott Lieb's many major accomplishments, there are numerous gems that sparkle even though they aren't among his most important.  In this note, we want to discuss something related to his beautiful note \cite{LiebClassLim} on the infinite spin limit of the pressure of quantum Heisenberg models which converges to a classical Heisenberg model, a work which motivated the second author's extension \cite{SimonClassLim} to more general Lie groups than $SU(2)$.  Lieb proved comparison inequalities for partition functions that squeezed the spin $S$ quantum partition function between the corresponding classical partition functions with slightly different coupling constants so that the difference of the coupling constants goes to zero as $S\to\infty$.  In the totally anisotropic case (where only $z$ components are coupled), Dyson, Lieb and the second author (we never published this work done in 1976 but it was included it in the 1993 book of Simon \cite[Section II.9]{SMLG}) proved monotonicity (increasing) of the partition function in $S$ and decreasing monotonicity if the coupling is scaled properly.

When we were working on Thomas-Fermi, Elliott taught the second author about the magic of convergence of convex functions - that convergence of convex functions implies convergence of derivatives at points where the limit is differentiable.  This implies convergence of certain correlation functions in the context that Lieb studied in \cite{LiebClassLim}.  But one loses for correlation functions inequalities like those that Dyson, Lieb and Simon found for partition functions.  It is that question that we want to discuss here.  Our framework will be less general in that we will only consider one component spins and more general in that we will allow general (even) apriori measures.

The second author is writing a book for Cambridge Press entitled \textit{Phase Transitions in the Theory of Lattice Gases} \cite{PTLG}. It is in many ways the successor to the 1993 book \cite{SMLG} from Princeton University Press.  That earlier book was mainly framework and largely left out all the most fun and beautiful elements of the theory: Correlation Inequalities, Lee-Yang, Peierls' Argument, Berezinskii-Kosterlitz-Thouless transitions and Infrared Bounds which are the subjects of the new book.  But since a different publisher is used, this is certainly \textit{not} volume 2 of the earlier work.

The framework for much of the subject is to fix a finite set $\Lambda\subset\bbZ^\nu$, and an apriori \textit{even} probability measure, $d\mu$, on $\bbR$, certainly with all moments finite and typically of compact support.

One considers the configurations in $\Lambda$, i.e. points $\mathbf{\sigma}$ in $\bbR^\Lambda$, indicated by $\{\sigma_j\}_{j\in\Lambda}$ and uncoupled measure with expectation
\begin{equation}\label{1.1}
  \jap{f}_{\mu,0} = \int f(\mathbf{\sigma})\,\prod_{j\in\Lambda} d\mu(\sigma_j)
\end{equation}
and one fixes a ferromagnetic Hamiltonian (i.e. $J(A) \ge 0$)
\begin{equation}\label{1.2}
  -H = \sum_{A\subset\Lambda} J(A) \sigma^A \qquad \sigma^A = \prod_{j\in A} \sigma_j
\end{equation}
or more general over mutliindices, i.e. assignments of an integer, $n_j\ge 0$ with then $\sigma^A = \prod_{j\in A} \sigma_j^{n_j}$ (and a finite sum or else $\ell^1$ condition).  One then considers, the Gibbs state
\begin{equation}\label{1.3}
   \jap{f}_{\mu,\Lambda} =Z^{-1} \jap{f e^{-H}}_{\mu,0};\qquad Z=\jap{e^{-H}}_{\mu,0}
\end{equation}

One studies the infinite volume limit with translation invariant $J(A)$, typically by proving stuff about the finite volume expectations.  The traditional case is the Ising model (aka spin $1/2$ Ising model) where $d\mu$ is a measure supported on $\pm 1$ each point with weight $1/2$; more generally, we'll refer to $b_T$ with weights $1/2$ at $\pm T$ ($b$ is for Bernoulli).   While a lot of the literature is specific to the spin $1/2$ Ising model, there is considerable, mathematically interesting, literature on more general (even) apriori measures.  Traditionally, one mainly considered the spin $S$ measure (for $S=\tfrac{1}{2},1,\tfrac{3}{2},....$, the measure with $2S+1$ pure points equally spaced symmetrically about $0$ and with equal weights) but the work of Guerra, Rosen and Simon \cite{GRS} and Griffiths-Simon \cite{GriffSi} on discrete approximations to Euclidean Quantum Field Theory changed that.

As the second author began to write about correlation inequalities in his new book, he wondered about a natural question.  We say that an apriori measure, $\nu$, on $\bbR$ \textit{Ising dominates} another measure $\mu$ if and only if for all $J(A)\ge 0$ and all $B$, one has that
\begin{equation}\label{1.4}
  \jap{\sigma^B}_{\mu,\Lambda} \le \jap{\sigma^B}_{\nu,\Lambda}
\end{equation}
In particular, for general $\mu$ compact support, does one have that $\mu$ Ising dominates $b_{T_-}$ and is Ising dominated by $b_{T_+}$ for suitable $0< T_- < T_+ <\infty$.  That would imply phase transitions occur for one apriori measure if and only if they do for all and inequalities on transition temperatures.

For most, even minor, aspects of the subject of correlation inequalities there are several papers, sometimes even dozens.  So it is surprised that the second author was unable to find a single published paper on the subject of what we just called Ising domination! Of course, it was unclear how to search for the subject in Google.  Eventually, we did find one 1978 paper of van Beijeren and Sylvester \cite{vBS} that we'll mention later (see Remark 2 after Theorem \ref{T3.1}) although in one respect it is unsatisfactory.  And we did also find an appendix of a paper on another subject but that gets ahead of our story (see the Remark 1 after Theorem \ref{T3.1}).

One of the pleasant things about writing a book on a subject that one once knew more about is that one gets to rediscover things that they have forgotten.  With the question of Ising domination in the back of his mind, the second author found an interesting footnote in a 1980 paper of Aizenman and er, B. Simon \cite{AiSiRotor}.  The footnote said

\medskip

\begin{quote}
  then by results of Wells (D. Wells, \textit{Some moment inequalities for general spin Ising ferromagnets}, Indiana Univ. preprint) $\jap{s_js_k}_{\beta,1} \le 2\jap{\sigma_j^{(1)}\sigma_k^{(1)}}_{\beta,2}$.
\end{quote}

\medskip

The left hand side is an Ising expectation and the right with the apriori measure of the $2D$ rotor with only couplings of the $1$ components.  So this was part of what seems to be an Ising domination result (the subscript $2$ indicates the Ising measure should really be $b_{1/\sqrt{2}}$).

So the second author set about finding this preprint.  Google didn't help directly but did point him to a 1984 paper of Chuck Newman that mentioned Wells' Indiana University PhD. thesis.  He wrote to Michael asking if he knew anything about our footnote and cced Chuck (who had been a grad student with the second author at Princeton) because the second author conjectured Wells had been his student.  Chuck replied and said he remembered that Wells had been Slim Sherman's student.  Sherman, the S of GKS and GHS was a delightful character, long dead.

So the second author wrote to Kevin Pilgrim, the chair at Indiana, who located a copy of Wells thesis \cite{WellsTh} on Proquest.  So far though, no one has had any luck on the preprint nor on locating Wells through Indiana University alumni records (but see later)!  While the thesis did not have anything directly about the above inequality, it did have a general framework on what we called the Ising domination problem, lovely material that should have been published.  After an initial draft of this note was written, the first two authors got some help and located Dr. Wells who kindly agreed to be a coauthor which makes sense since much of this Note publishes for the first time results from his thesis.  Also, we convinced him to allow us to continue to use the term \textit{Wells domination}.

Our main goal in this Note is to describe Wells' framework in Section \ref{s2} and what we regard as his most significant theorem in Section \ref{s3}.  Since Wells extended a framework of Ginibre, we begin Section \ref{s2} by reminding (telling) you of that. Then the notion we call Wells' domination followed by his big theorem in Section \ref{s3} and the notion of canonical lower bound.  We'll note there that the approach of van Beijeren and Sylvester \cite{vBS} has one big flaw in that there is no analog of the Wells Comparison Theorem, Theorem \ref{T3.1}.  Section \ref{s3A} will then make explicit the bounds on transition temperatures implied by Wells Comparison Theorem, recall a result of Griffiths \cite{GriffTrick} on comparison of transition temperatures for different spins and note that the bounds when a measure is canonical are equalities in mean field theory and so optimal in the high dimension limit. Section \ref{s4} has one of our two new results here - that the distribution of one component of an $D$-vector spin has a canonical lower bound.  From the footnote in \cite{AiSiRotor}, it is clear that the missing Wells preprint has the case $D=2$.  We'll see that case is much easier to prove than $D\ge 3$. In Section \ref{s5}, we will prove our most significant new result that the spin $S$ measures for $S\ne 1$ have a canonical lower bound and its consequence for improving Griffiths bound on transition temperatures for arbitrary spins.  Finally, an appendix \ref{App} provides the proof of a technical inequality.

We would especially like to thank Kevin Pilgrim for his help in locating Wells' thesis \cite{WellsTh} and Leonard Schulman and Joshua David Paik for helping the first two authors make contact with the third. We should also like to thank Terry Tao. The inequality \eqref{5.6} was originally only conjectured by the second author because it would imply the result on spin $S$ ($S\ne 1$) being canonical.  The second author asked Terry if he could prove the conjecture and Terry passed it on to the first author (who proved it) thereby making a successful shidduch.

%%%%%%%%%%%%%%%%%%%%%%%%%%%%%%%%%%%%%%%%%%%%%%%%%%%%%%%%%%%%%%
\section{Wells Framework} \lb{s2}
%%%%%%%%%%%%%%%%%%%%%%%%%%%%%%%%%%%%%%%%%%%%%%%%%%%%%%%%%%%%%%

As mentioned, Wells' approach is a slight modification of Ginibre's approach to the proof of GKS inequalities for Ising type models.  In a remarkable 1970 paper \cite{GinibreGKS}, Jean Ginibre (who alas passed away in March of 2020 at age 82) not only found a really simple proof of GKS inequalities but showed somewhat surprisingly that they held for all (even) apriori measures. If you are new to Ising models and have time for only one result, this one might be what you should know.

A \textit{Ginibre system} is a triple $\jap{X, \mu, \calF}$  of a compact Hausdorff space, $X$, a probability measure, $\mu$, on $X$ (with expectations $\jap{\cdot}_\mu$) and a class of continuous real valued functions $\calF\subset C(X)$ that obeys:
\begin{equation}\label{2.1}
  (G1) \qquad \forall_{f_1,...f_n\in\calF} \int_X f_1(x)\dots f_n(x)\,d\mu(x) \ge 0
\end{equation}
\begin{equation}\label{2.2}
  (G2) \qquad \forall_{f_1,...f_n\in\calF} \int_{X\times X} \prod_{j=1}^{n} \left(f_j(x)\pm f_j(y)\right) \, d\mu(x) d\mu(y) \ge 0
\end{equation}
for all $2^n$ choices of the plus and minus sign.

When it is clear which measure is intended, we will drop the $\mu$ from $\jap{\cdot}_\mu$. We have restricted to compact Hausdorff spaces and so bounded functions for simplicity.  But since all the arguments are essentially algebraic, all results extend to the case where $X$ is only locally compact so long as all $f\in\calF$ obey $\int |f(x)|^m\,d\mu(x) < \infty$ for all $m$ since that condition assures that all integrals below are convergent.

Note that
\begin{equation*}
  (G2) \Rightarrow 2\jap{f}_\mu = \int_X (f(x)+ f(y))\,d\mu(x)d\mu(y) \ge 0
\end{equation*}
and
\begin{align*}
  \int_{X\times X} (f(x) -& f(y))(g(x)-g(y)) \, d\mu(x) d\mu(y) \\
                          &= 2\left[\jap{fg}_\mu-\jap{f}_\mu \jap{g}_\mu\right]\ge 0
\end{align*}
We will see shortly that $(G2)\Rightarrow(G1)$.  What makes the notion so powerful is that there are three theorems for getting new Ginibre systems from old ones.

Given a family of functions, $\calF\subset C(X)$, we define the \textit{Ginibre cone}, $\calC(\calF)$, as the set of linear combinations with non-negative coefficients of products of functions from $\calF$.
\begin{theorem} [Ginibre Theorem 1] \lb{T2.1} If a triple $\jap{X, \mu, \calF}$ obeys $(G2)$, so does $\jap{X, \mu, \calC(\calF)}$.
\end{theorem}
It is trivial that $(G2)$ holds for sums and positive multiples of functions for which it holds, so it suffices to prove it holds for products.  By induction, we need only handle products of two functions.  We note that
\begin{equation}\label{2.3}
  fg\pm f'g' = \tfrac{1}{2}(f+f')(g\pm g')+\tfrac{1}{2}(f-f')(g\mp g')
\end{equation}
which allows us to prove $(G2)$ for a single product when we have it for individual functions (and shows (G2)$\Rightarrow$(G1)).

The following is trivial
\begin{theorem} [Ginibre Theorem 2] \lb{T2.2} Let $\left\{\jap{X_j, \mu_j, \calF_j}\right\}_{j=1}^n$ be a family of Ginibre systems.  Then $\jap{\times_{j=1}^n X_j,\otimes_{j=1}^n \mu_j, \cup_{j=1}^n \calF_j}$ is also a Ginibre system
\end{theorem}

And to add interactions, we use
\begin{theorem} [Ginibre Theorem 3] \lb{T2.3} Let $\jap{X, \mu, \calF}$ be Ginibre system.  Let $-H\in\calF$ and define a new measure, $\mu_H$ by
\begin{equation}\label{2.4}
  \jap{f}_{\mu_H} = \frac{\jap{fe^{-H}}_\mu}{\jap{e^{-H}}_\mu}
\end{equation}
Then $\jap{X, \mu_H, \calF}$ is a Ginibre system.
\end{theorem}

The proof is easy.  The normalization is irrelevant and we expand the exponential $\exp(-H(x)-H(y))$.  Finally
\begin{theorem} [Ginibre Theorem 4] \lb{T2.4} Let $X$ be $\bbR$ or a compact subset of the form $[-A,A]$ and let $d\mu$ be a probability measure which is invariant under $x\mapsto -x$ and so that (only non-trivial in case $X$ is not compact) $\int x^{2n} \, d\mu(x) < \infty$ for all $n$.  Let $\calF$ contain the single function, $f(x)=x$.  Then $\jap{X, \mu, \calF}$ is a Ginibre system.
\end{theorem}

The proof is easy! $(G2)$ says that for all non-negative integers, $k$ and $m$, one has that
\begin{equation}\label{2.5}
  \int_{X\times X} (x + y)^k(x - y)^m \, d\mu(x) d\mu(y) \ge 0
\end{equation}
Interchanging $x$ and $y$ implies the integral is zero if $m$ is odd and $(x,y)\mapsto (-x,-y)$ symmetry implies the integral is zero if $m+k$ is odd. Thus the only possible non-zero integrals are when $m$ and $k$ are even in which case the integrand is positive!

A little thought shows that for Hamiltonians of the form
\begin{equation}\label{2.6}
  -H = \sum_{A\subset\Lambda} J(A) \sigma^A \qquad \sigma^A = \prod_{j\in A} \sigma_j
\end{equation}
with \textit{any} (!!!) even apriori measure, one has positive expectations and positive correlations of the $\sigma^A$ which is GKS inequalities for general even measures.

We'd be remiss if we left the subject Ginibre's wonderful paper without mentioning two other examples he gives of Ginibre systems that are not relevant to Wells, although one will appear later. The first is to note that he proves that if $d\mu$ is a product of rotation invariant measures on circles, the set of functions $\cos(\sum_{j=1}^{n}m_j\theta_j)$ is a Ginibre system. This and some extensions are essentially half the correlation inequalities for plane rotors.

The second is related to an 1882 paper of Chebyshev \cite{Cheb} (which I don't think Ginibre knew about when he wrote his 1970 paper) which contained what is probably the earliest correlation inequality: Chebyshev proved that if $f, g$ are two monotone functions on $[0,1]$, then
\begin{equation}\label{2.7}
  \int_{0}^{1} f(x)g(x)\,dx \ge \int_{0}^{1} f(x)\,dx \int_{0}^{1} g(x)\,dx
\end{equation}
Ginibre proved that for any (not necessarily even) positive probability measure on $\bbR$, the set $\calF$ of all positive monotone functions is a Ginibre family.  The proof is again very easy.  This is a sort of poor man's FKG inequalities.

This completes our review of Ginibre, so we turn to Wells' work.  There is a simple extension of Ginibre's method in Wells' thesis \cite{WellsTh} that allows comparison of measures. Given two probability measures, $\mu$ and $\nu$ on a locally compact space, $X$, we say that $\mu$ \textit{Wells dominates} $\nu$, written  $\mu\triangleright\nu$ or $\nu\triangleleft\mu$ with respect to a class of continuous functions $\calF$ (with all moments of all $f\in\calF$ finite with respect to both measures; not needed if $X$ is compact) if for all $n$ and all $f_1, f_2,\dots,f_n$ and all $2^n$ choices of $\pm$, we have that
\begin{equation}\label{2.8}
  \int \int (f_1(x)\pm f_1(y))\dots (f_n(x)\pm f_n(y)) d\mu(x) d\nu(y) \ge 0
\end{equation}

We will be most interested in case $X=\bbR$, $\mu$ and $\nu$ are both even measures with all moments finite and $\calF$ has the single function $f(x)=x$ in which case the condition takes the form %
\begin{equation}\label{2.9}
  \int_\bbR \int_\bbR (x+y)^n (x-y)^m    d\mu(x) d\nu(y) \ge 0
\end{equation}
for all non-negative integers, $n$ and $m$ in which case we use the symbol $\triangleleft$ without being explicit about $\calF$.  Since the measures are even, one need only check this when $n+m$ is even.  It is trivial if both are even, so we only need worry about the case that both are odd. Since the measures are different, we don't have the exchange symmetry that makes the integral vanish if both are odd but symmetry under $y\mapsto -y$ implies invariance under interchange of $m$ and $n$, so we need only check for $m\ge n$.  We'll see examples later.

Extending the Ginibre machine is effortless.  It is easy to prove that
\begin{theorem} [Wells \cite{WellsTh}] \lb{T2.5} (a) If $\mu\triangleleft\nu$ for a set of functions $\calF$, the same is true for the Ginibre cone $\calC(\calF)$.

(b) If for $j=1,\dots,n$, $\mu_j\triangleleft\nu_j$ for probability measures on spaces $X_j$ with respect to sets of functions $\calF_j$ on $X_j$, then for the measures on $\prod_{j=1}^{n}X_j$ and the set of functions $\cup_{j=1}^n\calF_j$, one has that $\otimes_{j=1}^n\mu_j\triangleleft\otimes_{j=1}^n\nu_j$.

(c) If $\mu\triangleleft\nu$ for probability measures on a space $X$ with respect to a set of functions $\calF$ on $X$, if $-H\in\calF$ and if $\mu_H$, $\nu_H$ are Gibbs measures, then $\mu_H\triangleleft\nu_H$ for $\calF$.

(d) If $\mu\triangleleft\nu$ with respect to a set of functions $\calF$, then for every $f\in\calF$, we have that
\begin{equation}\label{2.10}
  \int f(x)\, d\mu(x) \le \int f(x)\, d\nu(x)
\end{equation}
\end{theorem}

This immediately implies that

\begin{corollary} [Wells \cite{WellsTh}] \lb{C2.6} If for $j=1,\dots,n$, $\mu_j\triangleleft\nu_j$ for probability measures on spaces $X_j$ with respect to sets of functions $\calF_j$ on $X_j$,then if $-H\in\calC(\cup_{j=1}^n\calF_j)$ and if $\mu_H, \nu_H$ are formed from the underlying product measures $\otimes_{j=1}^n\mu_j$ and $\otimes_{j=1}^n\nu_j$, then for all $F\in\calC(\cup_{j=1}^n\calF_j)$, one has that $\int f(x)\, d\mu_H(x) \le \int f(x)\, d\nu_H(x)$. In particular, if each $X_j=\bbR$, (so implicitly $F_j$ is the single function $\sigma_j$) and if $H$ has the general ferromagnetic Ising form, \eqref{2.6} with all $J(A)\ge 0$, then for all $A\subset 2^{\{1,\dots,n\}}$ one has that
\begin{equation}\label{2.11}
  \jap{\sigma^A}_{\mu_H} \le \jap{\sigma^A}_{\nu_H}
\end{equation}
\end{corollary}

Thus by the definition, \eqref{1.4}, of Ising domination, we see that if $\mu\triangleleft\nu$, then $\nu$ Ising dominates $\mu$.

Of course, $\triangleleft$ is a binary relation and it is tempting to think of it as a partial order on measures on $\bbR$ with all moments finite. Indeed, it is certainly reflexive. It is almost antisymmetric.  It is easy to see that $\mu\triangleleft\nu$ and $\nu\triangleleft\mu$ if and only if $\mu$ and $\nu$ have the same moments.  Thus it is antisymmetric among the measures of compact support or among measures obeying $\int e^{Ax^2}\, d\mu(x)<\infty$ for some $A>0$ but not among all measures with finite moments because of the possibilities of measures non-unique for the moment problem. But we do not know the following

\textbf{Question 1} Is Wells relation transitive among all even measures on $\bbR$?  How about among all measures on a general topological space if $\calF$ is rich enough?

Since Ising domination is trivially transitive, for applications, this lack isn't so important.

%%%%%%%%%%%%%%%%%%%%%%%%%%%%%%%%%%%%%%%%%%%%%%%%%%%%%%%%%%%%%%
\section{The Wells Comparison Theorem} \lb{s3}
%%%%%%%%%%%%%%%%%%%%%%%%%%%%%%%%%%%%%%%%%%%%%%%%%%%%%%%%%%%%%%

Given an even measure probability, $\mu$ on $\bbR$ and $s>0$, we define its scaling by
\begin{equation}\label{3.1}
  \mu^{(s)}[A] = \mu[s^{-1}A]
\end{equation}
Then the Bernoulli measure $b_S$ defined after \eqref{1.3} obeys $b_S= (b_1)^{(S)}$.

Even if it is not true that $\nu\triangleleft\mu$, it can happen that $\nu^{(s)}\triangleleft\mu$ for $s$ sufficiently small.  In the next section, we will see that this implies a bound on transition temperatures, so such comparison results are interesting.  The main result of this section implies that any two non-trivial measures of compact support are comparable in this sense.  Well's most important result is

\begin{theorem} [Wells Comparison Theorem \cite{WellsTh}] \lb{T3.1} Let $d\mu$ be an even probability measure on $\bbR$ with compact support that is not a point mass at $0$.  Then there are two strictly positive numbers, $T_-(\mu)$ and $T_+(\mu)$, so that $\mu\triangleleft b_S$ if and only if $S\ge T_+$ and $b_S\triangleleft\mu$ if and only if $S\le T_-$.  Moreover
\begin{equation}\label{3.2}
  T_+ = \sup\{s\,\mid\, s\in\supp(\mu)\}
\end{equation}
and
\begin{equation}\label{3.3}
  S\le T_-\iff \forall_{n\in\bbN} \int_\bbR (x^2-S^2)^n\,d\mu(x) \ge 0
\end{equation}
\end{theorem}

\begin{remarks} 1. Bricmont-Lebowitz-Pfister \cite{BLP} state the existence of $T_-$ part of this theorem, quoting Wells and providing his proof.

2.  There is a very different order from Wells order defined by van Beijeren-Sylvester \cite{vBS} (discussed further in \cite[Section 2.2]{PTLG}) that also implies Ising domination but it has the serious flaw that if $0\in\supp(\nu)$, then for no $T>0$ does $\nu$ dominate $b_T$ in their order.

3.  The proof below is essentially that of Wells.
\end{remarks}

\begin{lemma} \lb{L3.2} Let $\mu$ be a positive measure on an interval $I\subset\bbR$ (either open or closed at each endpoint).  Let $f,g\in L^2(d\mu)$ and suppose that $g$ is monotone increasing on $I$ and there is $c\in I$ so that $f(x)\le 0$ (resp $f(x)\ge 0$) if $x\le c$ (resp $x\ge c$). Then
\begin{equation}\label{3.4}
  \int f(x)g(x) \, d\mu(x) \ge g(c) \int f(x)\,d\mu(x)
\end{equation}
\end{lemma}

\begin{proof} The function $f(x)[g(x)-g(c)]$ is positive so its integral is positive which is the claim.
\end{proof}

\begin{proof} [Proof of Theorem \ref{T3.1}]  We first prove the existence of $T_+$ and \eqref{3.2}.  If $S\ge \sup\{s\,\mid\, s\in\supp(\mu)\}$, then, for the integrand in \eqref{2.9} to be positive, we need that $(S+y)^n(S-y)^m+(S+y)^m(S-y)^n\ge 0$ for all $y\ge 0$ in $\supp(\mu)$. If $\mu(\{0\})>0$, there is an additional term of $S^{n+m}\mu(\{0\})$ in the right hand side, but that is also positive, so for such $S$, we have that $\mu\triangleleft b_S$.

On the other hand, if $\mu\triangleleft b_S$, we have that $\int x^{2N}\,d\mu(x) \le S^{2N}$, so, taking $2N$th roots and then $N\to\infty$, we see that $S\ge \sup\{s\,\mid\, s\in\supp(\mu)\}$ which proves the formula for $T_+$.

Next we will prove that
\begin{equation}\label{3.5}
  b_S\triangleleft\mu\iff \forall_{n \text{ odd}}\int_\bbR (x^2-S^2)^n\,d\mu(x) \ge 0
\end{equation}
Taking $n=m$ in the basic integral, we see that
\begin{equation}\label{3.6}
  b_S\triangleleft\mu\Rightarrow \forall_{n \text{ odd}}\int_\bbR (x^2-S^2)^n\,d\mu(x) \ge 0
\end{equation}

Now look at the basic integral when $\nu=b_S$ and $m>n$ with both odd. Since $(x\pm S)^n(x\mp S)^m = (x^2-S^2)^n(x\mp S)^{m-n}$ we see that the integral in question is
\begin{align}
 \tfrac{1}{2} \int(x^2&-S^2)^n \left[(x+S)^{m-n}+(x-S)^{m-n}\right] \, d\mu(x)  \nonumber\\
      &=\int (x^2-S^2)^n \left[(x+S)^{m-n}+(x-S)^{m-n}\right] \, d\tilde{\mu}(x) \lb{3.7}
\end{align}
where $\tilde{\mu}$ is the measure restricted to $(0,\infty)$ plus $\tfrac{1}{2}\mu(\{0\})\delta_0$.  By the binomial theorem, the polynomial $Q_{2k}(y)=(y+S)^{2k}+(y-S)^{2k}$ only has even degree terms with only positive coefficients so the function in $[\cdot]$ in the last equation is monotone on $I=[0,\infty)$.  Applying the lemma with $c=S$, we see that
\begin{equation} \lb{3.8}
  \int_\bbR \int_\bbR (x+y)^n (x-y)^m    d\mu(x) d\nu(y) \ge (2S)^{m-n}\int_\bbR (x^2-S^2)^n\,d\mu(x)
\end{equation}
Thus, we have proven \eqref{3.5}.

Finally, we show that $T_->0$.  First, pick $a>0$ so that $\mu([a,\infty)) > 0$. Pick $0<b<a$ so small that
\begin{equation} \lb{3.9}
  \frac{b^2}{a^2-b^2} \le \min\left(1, 2\mu([a,\infty))\right)
\end{equation}
possible since the left side goes to zero as $b\downarrow 0$. Since the integrand is positive on $[b,a]$, we have that for all $k\in\bbN$
\begin{equation} \lb{3.10}
  \int (x^2-b^2)^{2k+1} d\mu(x) \ge -(b^2)^{2k+1} + 2(a^2-b^2)^{2k+1}\mu([a,\infty))
\end{equation}
\begin{equation} \lb{3.11}
     \null\qquad  = 2(a^2-b^2)^{2k+1}\left[2\mu([a,\infty))-\left(\frac{b^2}{a^2-b^2}\right)^{2k+1}\right] \ge 0
\end{equation}
by the choice of $b$.  Thus $T_-\ge b>0$.
\end{proof}

One consequence of the theorem is
\begin{equation}\label{3.12}
  T_- \le \left(\int_\bbR x^2\,d\mu(x)\right)^{1/2}
\end{equation}
It is an interesting question when one has equality.  One would like as good a lower bound on $T_-$ as possible which can yield good lower bounds on transition temperatures.  Often one has equality in \eqref{3.12} in which case we will say that $T_-$ is \textit{canonical} for $\mu$.

\begin{example} \lb{E3.3}  We consider spins taking three values.  For $0\le\lambda\le 1$, consider the probability measure supported by the three points $\{0,\pm1\}$ given by
\begin{equation}\label{3.13}
  d\mu_\lambda = \tfrac{\lambda}{2}\left(\delta_1+\delta_{-1}\right)+(1-\lambda)\delta_0
\end{equation}
For $\lambda=2/3$, which is equal weights, this is called (normalized) spin $1$. Then
\begin{align}\label{3.14}
  \jap{(x^2-T^2)^{2m+1}}_\lambda & =  (1-T^2)^{2m+1}\lambda-(1-\lambda)T^{2(2m+1)} \nonumber\\
                                 & \ge 0  \iff \left[\frac{1-T^2}{T^2}\right]^{2m+1} \ge \frac{1-\lambda}{\lambda} \nonumber \\
                                 & \iff \frac{1-T^2}{T^2} \ge \left(\frac{1-\lambda}{\lambda}\right)^{1/2m+1}
\end{align}
If $\lambda\le 1/2$, then $(1-\lambda)/\lambda\ge 1$ and the maximum on the right side of the last formula occurs for $m=0$ while, if $\lambda\ge 1/2$, then $(1-\lambda)/\lambda\le 1$ and we get the maximum as $m\to\infty$.  Thus, we find that
\begin{equation}\label{3.15}
  T_-(\lambda) = \left\{
                   \begin{array}{ll}
                     \sqrt{\lambda}, & \hbox{ if }\lambda\le\tfrac{1}{2} \\
                     \sqrt{\tfrac{1}{2}}, & \hbox{ if }\lambda\ge\tfrac{1}{2}
                   \end{array}
                 \right.
\end{equation}
So we see there are cases where $T_-= \jap{x^2}^{1/2}=\sqrt{\lambda}$ and other cases where the inequality is strict. Note also that at $\lambda=1/2$, the integral $\jap{(x^2-T_-^2)^{2m+1}}_\lambda$ vanishes for all $n$, a sign that the distribution of $x^2-T_-^2$ is symmetric about $0$.
\end{example}

In the remainder of this Note, we will discuss bounds on transition temperatures and then two interesting classes: in Section \ref{s4}, the distribution of a single component of a $D$-vector model and in Section \ref{s5}, the spin $S$ spin.

%%%%%%%%%%%%%%%%%%%%%%%%%%%%%%%%%%%%%%%%%%%%%%%%%%%%%%%%%%%%%%
\section{Bounds on Transition Temperatures} \lb{s3A}
%%%%%%%%%%%%%%%%%%%%%%%%%%%%%%%%%%%%%%%%%%%%%%%%%%%%%%%%%%%%%%

Fix a translation invariant ferromagnetic interaction, $J(i-j)\ge 0$ and an even apriori measure, $\mu$.  Let $T_c(\mu)$ be the transition temperature for the model defined as the unique temperature (which may be zero if there is no phase transition!) so that for larger temperatures, the two point infinite volume free boundary condition state $\jap{\cdot}$ has
\begin{equation}\label{3A.1}
  \lim_{j\to\infty} \jap{\sigma_j\sigma_0} = 0
\end{equation}
We want to see what $\nu^{(s)}\triangleleft\mu$ implies about the relation of $T_c(\mu)$ and $T_c(\nu)$ (a similar analysis holds with other possible definitions of transition temperature).

The arguments below while stated for Wells order only depend on Ising domination.  Making the temperature and measure explicit, with $\jap{\cdot}_{T,\mu}$ the infinite volume free BC state, we note that by the definition of $\nu^{(s)}$, we have that (because we are assuming only pair interactions and because temperature appears as $\sigma_i\sigma_j/T$)
\begin{equation}\label{3A.2}
  \jap{\sigma^A}_{T,\nu^{(s)}} = s^{|A|}\jap{\sigma^A}_{T/s^2,\nu}
\end{equation}
Thus, since Wells order implies Ising domination, we see that $\nu^{(s)}\triangleleft\mu$ implies that
\begin{equation}\label{3A.3}
  \jap{\sigma^A}_{T,\mu} \ge s^{|A|}\jap{\sigma^A}_{T/s^2,\nu}
\end{equation}

Therefore, if $T\le s^2T_c(\nu)$, we see that \eqref{3A.1} fails for $\jap{\sigma_j\sigma_0}_{T/s^2,\nu}$ and so by \eqref{3A.3} for $\jap{\sigma_j\sigma_0}_{T,\mu}$. We have thus proven that

\begin{proposition} \lb{P3A.1} Let  $\nu^{(s)}\triangleleft\mu$ for two non-trivial even measures.  Then
\begin{equation}\label{3A.4}
  T_c(\mu) \ge s^2 T_c(\nu)
\end{equation}
In particular
\begin{equation}\label{3A.5}
  T_c(\mu) \ge T_-(\mu)^2 T_c(\text{\textit{classical Ising}})
\end{equation}
so if $T_-$ is canonical for $\mu$, then
\begin{equation}\label{3A.6}
  T_c(\mu) \ge \jap{x^2}_\mu T_c(\text{\textit{classical Ising}})
\end{equation}
\end{proposition}

This last putative inequality is especially interesting because the mean field transition temperature (see, for example, \cite[Section II.13]{SMLG} or \cite[Section 2.6]{PTLG}) is given by
\begin{equation}\label{3A.7}
  T_{MF}(\mu) = \jap{x^2}_\mu \sum_{j} J(j)
\end{equation}
for a pair interacting ferromagnetic model.  Thus one has equality in \eqref{3A.6} if $T_c$ is replaced by $T_{MF}$.  It is known \cite{BKLS, PTLG} that mean field theory is exact in the infinite dimensional limit of nearest neighbor generalized Ising models (in the sense that $T_c/T_{MF}\to 1$ as $d\to\infty$ for the model on $\bbZ^d$).  \cite{BKLS} only discusses spin $1/2$ but to get that equality holds in \eqref{3A.6}, it suffices to get a MF lower bound on $T_c(\mu)$ and \cite{BKLS} get that from Fr\"{o}hlich-Simon-Spencer \cite{FSS} whose argument works for any spin (for many models of interest including those of the next two sections, there is also a mean field upper bound on transition temperatures - see \cite[Section 2.6]{PTLG} and \cite{Pearce, Tasaki}).  In event, we see that \textit{if \eqref{3A.6} holds for all ferromagnetic pair interactions, then the constant $\jap{x^2}_\mu$ is best possible}.  A major theme of the rest of the paper is proving \eqref{3A.6} in two classes of models.

One of these is the spin $S$ measure and we want to end this section by noting what we believe is the best prior lower bound on $T_c(\tilde{\mu}_S)$.  To be explicit, for each value of $S=1/2,1,3/2,...$, consider the measure $\tilde{\mu}_S$ which takes $2S+1$ values equally spaced between $-1$ and $1$, each with weight $1/(2S+1)$.  It is interesting to find the square, $T_-(S)^2$, of the Wells $T_-$ associated to $\tilde{\mu}_S$ because if $T_c(S)$ is the transition temperature for a model with apriori measure $\tilde{\mu}_S$ and some fixed two point ferromagnetic interaction, then Proposition \ref{P3A.1} (and the easy $T_+(S)=1$) implies that
\begin{equation}\label{3A.8}
  T_-(S)^2  T_c\left(\tfrac{1}{2}\right) \le T_c(S) \le  T_c\left(\tfrac{1}{2}\right)
\end{equation}
(we'll prove in Section \ref{s5} that $\tilde{\mu}_S$ is canonical and compute $\jap{x^2}_{\mu_S}$). So far as we know the best previous result of this genre in the literature is due to Griffiths \cite[eq.(4.23)]{GriffTrick} who proved that

\begin{theorem} [Griffiths \cite{GriffTrick}] \lb{T3A.2} One has that
\begin{equation}\label{3A.9}
   \tfrac{1}{4}  T_c\left(\tfrac{1}{2}\right) \le T_c(S) \le  T_c\left(\tfrac{1}{2}\right)
\end{equation}
\end{theorem}

\begin{remarks}  1. This is what Griffiths proves for $S$ an integer.  For $2S=2k+1$ odd, he proves the slightly strong result with $\tfrac{1}{4}$ replaced by $(k+1/2k+1)^2$. Please note that what we call $T_c(S)$, Griffiths denotes $T_c(2S)$, e.g. $T_c(\text{\textit{classical Ising}})$ which we denote  $T_c\left(\tfrac{1}{2}\right)$, he denotes as $T_c(1)$.

2. This paper of Griffiths \cite{GriffTrick} is best known for proving Lee-Yang and GKS inequalities by realizing spin $S$ (normalized so the maximum value is $2S$, a measure we call $\mu_S$) by $2S$ spins with values $\pm 1$ with finite ferromagnetic couplings (given by Figs 2 and 3 in his paper) but he notes that one can also realize them with $S$ (if $S$ is an integer or $S+\tfrac{1}{2}$ if that's an integer) frozen together and then GKS implies that $\tilde{\mu}_S$ Ising dominates $b_{T=1/2}$.
\end{remarks}

%%%%%%%%%%%%%%%%%%%%%%%%%%%%%%%%%%%%%%%%%%%%%%%%%%%%%%%%%%%%%%
\section{Totally Anisotropic D-vector model} \lb{s4}
%%%%%%%%%%%%%%%%%%%%%%%%%%%%%%%%%%%%%%%%%%%%%%%%%%%%%%%%%%%%%%

We turn next to one of the two new results on this subject.  It involves the interesting measure
\begin{equation}\label{4.1}
  d\mu_D(x) =\left[\frac{\Gamma\left(\tfrac{D}{2}\right)}{\sqrt{\pi}\,\Gamma\left(\tfrac{D-1}{2}\right)}\right]
                   (1-x^2)^{\tfrac{1}{2}(D-3)}\chi_{[-1,1]}(x) dx
\end{equation}
This is the distribution of $x_1$ if one looks at a $D$-component unit vector, $\mathbf{x}=(x_1,\dots,x_D)$, distributed with the rotation invariant probability measure on $\bbS^{D-1}$.  Since, with respect to this measure, all $x_j$ have the same distribution and $\sum_{j=1}^{D} x_j^2=1$, we clearly have that
\begin{equation}\label{4.2}
  \jap{x^2}_D = 1/D
\end{equation}

\begin{theorem} \lb{T4.1} $T_-(\mu_D)$ is canonical, i.e. $T_-(\mu_D)^2=1/D$, so, in particular, \eqref{3A.6} holds for $\mu_D$.
\end{theorem}

The result for $D=2$ is especially easy because, for all $m$, $\jap{(x^2-1/2)^{2m+1}}_{D=2} = 0$. To prove this, note that it is equivalent to $\jap{(2x^2-1)^{2m+1}}_{D=2} = \jap{(x_1^2-x_2^2)^{2m+1}}_{\text{rotor}}=0$ which follows by $x_1\leftrightarrow x_2$.  I note that this result for $D=2$ is precisely the result that Aizenman and the second author say is in Wells' mystery preprint.

\begin{lemma} \lb{L4.2}  Let $\mu, \nu$ be two measures of compact support on $[0,\infty)$.  Then
\begin{equation}\label{4.3}
  \int f(x)\, d\mu(x) \le \int f(x) \, d\nu(x)
\end{equation}
for all monotone functions with $f(0)=0$ if and only if
\begin{equation}\label{4.4}
  \forall_{C>0}\, \mu([C,\infty)) \le \nu([C,\infty))
\end{equation}
\end{lemma}

\begin{proof} If \eqref{4.3} holds, we take $f$ to be the characteristic function of $[C,\infty)$ to get \eqref{4.4}.

Now suppose that we have \eqref{4.4}. By a simple approximation argument, it suffices to prove \eqref{4.3} for $f$'s which are $C^1$ with $f(0)=0$.  For such $f$'s, we have that
\begin{equation}\label{4.5}
  f(x) = \int_{0}^{\infty} f'(C)\chi_{[C,\infty)}(x)\,dC
\end{equation}
by doing the integral.  Thus
\begin{equation}\label{4.6}
  \int f(x)\, d\mu(x) = \int_{0}^{\infty} f'(C)\mu([C,\infty))\,dC
\end{equation}
$f$ monotone implies that $f'\ge 0$, so that \eqref{4.4} implies \eqref{4.3}.
\end{proof}

\begin{lemma} \lb{L4.3} Let $a,b>0$ and suppose that $g$ is positive on $(-b,a)$ with
\begin{equation}\label{4.7}
  \int_{-b}^{0} g(x)\, dx = \int_{0}^{a} g(x)\, dx; \qquad 0\le y< b \Rightarrow g(-y)\ge g(y)
\end{equation}
Let $f$ be an odd, monotone increasing, continuous function on $(-b,a)$. Then
\begin{equation}\label{4.8}
  \int_{-b}^{a} f(x)g(x)\, dx \ge 0
\end{equation}
\end{lemma}

\begin{remark} The result is quite intuitive. The condition on $g$ says that the measure $g(x)\,dx$ is concentrated on the right at larger $x$ than on the left, so more concentrated where $|f|$ is larger.
\end{remark}

\begin{proof} Define the measures on $[0,\infty)$:
\begin{equation}\label{4.9}
  d\nu(x) = \chi_{[0,a]}(x) g(x)\, dx; \qquad d\mu(x) = \chi_{[0,b]}(x) g(-x)\, dx
\end{equation}
We claim that for all $C>0$, one has that
\begin{equation}\label{4.10}
  \nu([0,C)) \le \mu([0,C))
\end{equation}
so, by \eqref{4.7}, we have that \eqref{4.4} holds.  Thus, since $f$ is monotone, by Lemma \ref{L4.2}, we have that
\begin{equation}\label{4.11}
  \int_{0}^{b} f(x) g(-x)\, dx \le \int_{0}^{a} f(x) g(x)\, dx
\end{equation}
which is \eqref{4.8} since $f$ is odd.  So we need only prove \eqref{4.10}.

If $C\le b$, \eqref{4.10} is immediate from the fact that for $x\ge 0$, $g(-x)\ge g(x)$.  By \eqref{4.7}, if $C>b$, then
\begin{equation}\label{4.12}
  \mu([0,C)) = \mu([0,\infty))= \nu([0,\infty)) \ge \nu([0,C))
\end{equation}
\end{proof}

\begin{proof} [Proof of Theorem \ref{T4.1}] As noted after the Theorem, $1/D$ is the second moment. Thus, by \eqref{3.12}, $T_-^2\le 1/D$, so, by \eqref{3.3}, we need only prove that for all $m\in\bbN$, the have that
\begin{equation}\label{4.13}
  \int_{-1}^{1}\left(x^2-\tfrac{1}{D}\right)^{2m+1} (1-x^2)^{\tfrac{1}{2}(D-3)}\,dx \ge 0
\end{equation}

Note first that, since the integrand is even in $x$, we can integrate only from $0$ to $1$ and then change variables from $x$ to $y=x^2-1/D$.  One sees that \eqref{4.13} is equivalent to
\begin{equation}\label{4.14}
  \int_{-1/D}^{1-1/D} \frac{y^{2m+1}\left(1-\tfrac{1}{D}-y\right)^{\tfrac{1}{2}(D-3)}}{\left(y+\tfrac{1}{D}\right)^{1/2}}\,dy \ge 0
\end{equation}

If $D=2$, then the integral goes from $-1/2$ to $1/2$ and the integrand is $y^{2m+1}(y^2-\tfrac{1}{4})^{-1/2}$ which is odd so the integral in \eqref{4.14} is $0$ for all $m$, recovering what we saw above. If $D\ge 3$, define, for $m\ge 1$
\begin{equation}\label{4.15}
  g(y) = \frac{|y|\left(1-\tfrac{1}{D}-y\right)^{\tfrac{1}{2}(D-3)}}{\left(y+\tfrac{1}{D}\right)^{1/2}}; \qquad f(y)=\sgn(y)y^{2m}
\end{equation}
The fact that $1/D$ is the second moment of $d\mu_D$ implies the integral in \eqref{4.13} vanishes if $m=0$ which means that $g$ obeys the first equation in \eqref{4.7}. Since $g(y)/|y|$ is monotone decreasing in $y$, the inequality in the second half of \eqref{4.7} holds. $f$ is odd and monotone increasing, so \eqref{4.8} implies \eqref{4.14}.
\end{proof}

%%%%%%%%%%%%%%%%%%%%%%%%%%%%%%%%%%%%%%%%%%%%%%%%%%%%%%%%%%%%%%
\section{Spin S} \lb{s5}
%%%%%%%%%%%%%%%%%%%%%%%%%%%%%%%%%%%%%%%%%%%%%%%%%%%%%%%%%%%%%%

For each value of $S=1/2,1,3/2,...$, consider the measure $\tilde{\mu}_S$ which takes $2S+1$ values equally spaced between $-1$ and $1$, each with weight $1/(2S+1)$.  We will prove here that except for $S=1$, this measure is canonical which will lead to improvements in the first inequality in \eqref{3A.9} by a factor of at least $\tfrac{4}{3}$ and which by discussion after \eqref{3A.7} yields optimal constants. We begin by computing $\jap{x^2}_{\tilde{\mu}_S}$

\begin{theorem} \lb{T5.1} We have that
\begin{equation}\label{5.1}
  a_S \equiv \int x^2 d\tilde{\mu}_S(x) = \frac{1}{3}\frac{S+1}{S}
\end{equation}
\end{theorem}

\begin{remarks} 1. Using $\sum_{j=1}^{S} j^2=\tfrac{S(S+1)(2S+1)}{6}$, one easily gets this result for $S$ integral.  One can also use this formula to get the result for half odd integral $S$ by using the fact that this case can be rewritten as a sum over odd integers between $1$ and $2S$ which one can realize as a sum over all integers minus the sum over even integers.  That calculation is  awkward and it is surprising it yields the same algebraic formula which is why we give the slick uniform proof below.

2. We will often scale the measure by a factor of $S$ so the $2S+1$ values are equally spaced from $-S$ to $S$ yielding a measure $\mu_S$.  If we interpret $\sum_{j=-S}^{S}$ as the sum over the $-S, -S+1,\dots,S-1, S$ (the usual meaning if $S$ is an integer but unusual if $S$ is half and odd integer), then \eqref{5.1} is equivalent to
\begin{equation}\label{5.2}
  A_S \equiv \int x^2 d\mu_S(x) = \frac{1}{2S+1}\sum_{-S}^{S} j^2 =   \frac{S(S+1)}{3}
\end{equation}
\end{remarks}

\begin{proof} Use $\ntr(\cdot)$ for the normalized trace on finite dimensional spaces, i.e. the average of diagonal elements or $\tr(\cdot)$ divided by the dimension. Consider what a physicist would call a quantum spin, $\mathbf{\sigma}$, of spin $S$ and a mathematician the generators of the irreducible representation of $SU(2)$ of dimension $2S+1$.  The operator $\sigma_z$ has eigenvalues $-S,-S+1,\dots,S-1,S$ and the Casimir operator has the form:
\begin{equation}\label{5.3}
  \sigma_x^2+\sigma_y^2+\sigma_z^2 = S(S+1)\bdone
\end{equation}
Thus
\begin{align}\label{5.4}
 A_S&=\ntr\left(\sigma_z^2\right) = \frac{1}{3}\ntr(\left(\sigma_x^2+\sigma_y^2+\sigma_z^2\right) \nonumber \\
    &= \frac{1}{3}\ntr\left(S(S+1)\bdone\right) = \frac{S(S+1)}{3}
\end{align}
by the symmetry of the three directions.  This proves \eqref{5.2}
\end{proof}

We saw in Example \ref{E3.3} that for $S=1$ ($\lambda=2/3$), one has that $T_-=\sqrt{\tfrac{1}{2}}<\sqrt{\tfrac{2}{3}}= \sqrt{a_{S=1}}$, so $\tilde{\mu}_S$ is \textit{not} canonical for that value of $S$.  The main technical result of this section, which we will prove in Appendix \ref{App}, is

\begin{theorem} \lb{T5.2} Let $S\ne 1$ be half a positive integer.  Then
\begin{equation}\label{5.6}
  \sum_{j=-S}^{S}(3j^2-S(S+1))^{2m+1}\ge 0
\end{equation}
\end{theorem}

We will also see that the opposite inequality is true when $S=1$.  By Theorem \ref{T3.1}, this implies that

\begin{corollary} \label{C5.3} For any $S\ne 1$, we have that $T_-$ is canonical for $\tilde{\mu}_S$.
\end{corollary}

\noindent which in turn, by \eqref{3A.8} implies that

\begin{corollary} \lb{C5.4} For $S\ne 1$, one has that
\begin{equation}\label{5.5}
  T_c(S) \ge \left(\frac{1}{3}+\frac{1}{3S}\right)  T_c\left(\tfrac{1}{2}\right)
\end{equation}
while for $S=1$ one has that
\begin{equation}\label{5.7}
  T_c(1) \ge \frac{1}{2} T_c\left(\tfrac{1}{2}\right)
\end{equation}
\end{corollary}

This improves Griffiths result, \eqref{3A.9} by a factor of at least $4/3$.  To be totally accurate, as we noted Griffith's lower bound (as proven, not as stated!) is slightly stronger than \eqref{3A.9} when $S$ is half an odd integer.  In particular for $S=\tfrac{3}{2}$ what \eqref{3A.9} has as $\tfrac{1}{4}$, he actually proves is $\tfrac{4}{9}$ so our results are improvement but for this  value only by a factor of $5/4$ rather than $4/3$.

One final remark.  In the notation used in \eqref{5.4}, \eqref{5.6} is equivalent to
\begin{equation}\label{5.8}
  \tr\left((2\sigma_z^2-\sigma_x^2-\sigma_y^2)^{2m+1}\right) \ge 0
\end{equation}
We wonder if that couldn't be used for a matrix theoretic version of what we prove in the Appendix.

\appendix
%%%%%%%%%%%%%%%%%%%%%%%%%%%%%%%%%%%%%%%%%%%%%%%%%%%%%%%%%%%%%%
\section{A Majorization Bound} \lb{App}
%%%%%%%%%%%%%%%%%%%%%%%%%%%%%%%%%%%%%%%%%%%%%%%%%%%%%%%%%%%%%%

In this appendix, we will prove a general set of inequalities that includes \eqref{5.6}. We will prove two theorems

\begin{theorem} \lb{TA.1} Fix an integer $N\ge 1$, a function, $\psi$ on $[0,1]$, which is non-negative, continuous, strictly monotone increasing and convex and a function, $\Phi$, on $[-\Norm{\psi}_\infty,\Norm{\psi}_\infty]$ which is continuous, odd and whose restriction to $[0,\Norm{\psi}]$ is convex.  Let
\begin{equation}\label{A.1}
  \overline{\psi} = (N+1)^{-1}\sum_{j=0}^{N} \psi\left(\tfrac{j}{N}\right)
\end{equation}
Then
\begin{equation}\label{A.2}
  \sum_{j=1}^{N} \Phi\left(\psi\left(\tfrac{j}{N}\right)-\overline{\psi}\right) \ge 0
\end{equation}
\end{theorem}

\begin{remark} By translation and scaling, this result can easily be generalized.  For example, while stated for $N+1$ equally spaced points between $0$ and $1$, we will apply it to $N+1$ odd integers stating at $1$.  The map $k\mapsto (k-1)/2N$ maps those $N+1$ integers into the points of the theorem.
\end{remark}

\begin{theorem} \lb{TA.2} For an integer $N\ge 2$ and an even, continuous, convex function, $\psi$ on $[-1,1]$ and a function, $\Phi$, on $[-\Norm{\psi}_\infty,\Norm{\psi}_\infty]$ which is continuous, odd and whose restriction to $[0,\Norm{\psi}]$ is convex.  Let $\overline{\psi}$ be given by
\begin{equation}\label{A.3A}
  \overline{\psi} = (2N+1)^{-1}\sum_{j=-N}^{N} \psi\left(\tfrac{j}{N}\right)
\end{equation}
Suppose that
\begin{equation}\label{A.3B}
  2\psi(1)+\psi(0)+2\psi\left(\tfrac{1}{N}\right) \ge 5 \overline{\psi}
\end{equation}
If $N$ is odd, suppose that
\begin{equation}\label{A.3C}
  \psi\left(\tfrac{1}{2}+\tfrac{1}{2N}\right) \le \overline{\psi}
\end{equation}
Then
\begin{equation}\label{A.3}
  \sum_{j=-N}^{N} \Phi\left(\psi\left(\tfrac{j}{N}\right)-\overline{\psi}\right) \ge 0
\end{equation}
\end{theorem}

\begin{remarks} 1. We will see later (see \eqref{A.1B}) that $2\psi(1)+2\psi(0) \ge 4 \overline{\psi}$ which puts \eqref{A.3B} in perspective.

2. It might be true that this theorem holds without the need for the condition \eqref{A.3B} but it holds in the case we need so we didn't try hard to eliminate it.  We will explain later why a naive extension of the proof of Theorem \ref{TA.1} doesn't work in a simple example and led to the extra condition.  We do note that \eqref{A.3B} is a restriction.  If we normalize $\psi$ by $\psi(0)=0, \psi(1)=1$, then in the limit as $N\to\infty$, \eqref{A.3B} becomes
\begin{equation}\label{A.3D}
  \int_{0}^{1} \psi(x)\,dx \le \tfrac{2}{5}
\end{equation}
which for $\psi(x)=|x|^p$ requires $p\ge \tfrac{3}{2}$ while convexity only requires $p\ge 1$.

3. On the other hand, \eqref{A.3C}, as we will see, is quite natural independent of our method of proof.
\end{remarks}

One key idea behind the proofs is the theory of majorization.  We let $\bbR^n_{+,\ge}$ denote the set of $n$-tuples of reals, $\mathbf{x}=(x_1,\dots,x_n)$ with
\begin{equation}\label{A.4}
  x_1\ge x_2\ge\dots x_n \ge 0
\end{equation}
If $\mathbf{x}, \mathbf{y}\in\bbR^n_{+,\ge}$, we say that $\mathbf{x}$ \textit{majorizes} $\mathbf{y}$, written $\mathbf{x}\succ \mathbf{y}$ or $\mathbf{y}\prec \mathbf{x}$ if an only if
\begin{equation}\label{A.5}
  \sum_{j=1}^{n} x_j = \sum_{j=1}^{n} y_j; \quad S_k(\mathbf{x})\equiv\sum_{j=1}^{k} x_j \ge \sum_{j=1}^{k} y_j,\,  k=1,\dots,n-1
\end{equation}
which defines $S_k(\mathbf{x})$. A standard reference is Marshall-Olkin \cite{MarOlk} which has been called a love poem to majorization; other references are Hardy-Littlewood-P\'{o}lya \cite{HLP2} and Simon \cite[Chapters 14-15]{SimonConvex}.  We will rely on the following aspect of the theory:

\begin{theorem} [Karamata's Inequality] \lb{TA.3} Let $\mathbf{x}, \mathbf{y}\in\bbR^n_{+,\ge}$ with $\mathbf{x}\succ \mathbf{y}$ and let $\varphi$ be an arbitrary continuous convex function on $[0,x_1]$.  Then
\begin{equation}\label{A.6}
  \sum_{j=1}^{n} \varphi(x_j) \ge \sum_{j=1}^{n} \varphi(y_j)
\end{equation}
\end{theorem}

\begin{remarks} 1.  Even though this is widely referred to as Karamata's inequality (e.g. Cvetkovski\cite{Cvet} or Wikipedia\cite{Wiki}) after Karamata's 1932 paper \cite{Kara}, it or theorems that imply it appear in a 1923 paper of Schur \cite{Schur} and a 1929 paper of Hardy-Littlewood-P\'{o}lya \cite{HLP1}.  That said, we note that \cite{HLP1} doesn't have a proof which may not have appeared until \cite{HLP2} in 1934 and that Karamata proved a converse, namely, if $\mathbf{x}, \mathbf{y}\in\bbR^n_{+,\ge}$ and \eqref{A.6} holds for all convex $\varphi$, then $\mathbf{x}\succ \mathbf{y}$.

2. The idea of one proof is quite simple (for details, see, for example, Simon \cite[Theroem 1.9]{SimonTrace} or Simon \cite[Theroem 15.5]{SimonConvex}): by slicing with specific hyperplanes, one proves that $\mathbf{y}$ is in the convex hull in $\bbR^n$ of the (at most) $n!$ points obtained from $\mathbf{x}$ by permuting the coordinates and then one notes the function $\mathbf{w}\mapsto \sum_{j=1}^{n} \varphi(w_j)$ is convex and permutation symmetric.
\end{remarks}

In the case of Theorem \ref{TA.1}, we will prove that $\mathbf{x}\succ \mathbf{y}$ using a simple observation:

\begin{proposition} \lb{PA.4} Suppose that $\mathbf{x}, \mathbf{y}\in\bbR^n_{+,\ge}$ with $\sum_{j=1}^{n} x_j = \sum_{j=1}^{n} y_j$ and that for some $\ell\in 2,\dots,n-1$, one has that
\begin{equation}\label{A.7}
  j< \ell \Rightarrow x_j > y_j \qquad \qquad j \ge \ell \Rightarrow x_j \le y_j
\end{equation}
Then $\mathbf{x}\succ \mathbf{y}$.
\end{proposition}

\begin{proof} If $k< \ell$, it is immediate that $\sum_{j=1}^{k} x_j \ge \sum_{j=1}^{k} y_j$ and similarly, it is immediate that if $k\ge\ell$, then  $\sum_{j=k}^{n} x_j \le \sum_{j=k}^{\ell} y_j$.  Subtracting this from $\sum_{j=1}^{n} x_j = \sum_{j=1}^{n} y_j$, we see that also for $k\ge\ell$, one has that $\sum_{j=1}^{k} x_j \ge \sum_{j=1}^{k} y_j$.
\end{proof}

We need two preliminaries for the proof of Theorem \ref{TA.1}:

\begin{lemma} \lb{LA.5} Let $\psi$ be a convex function on $[0,1]$ and suppose that
\begin{equation}\label{A.8}
  0\le \tilde{b}\equiv 2c-b < \tilde{a} \equiv 2c-a \le c \le a < b \le 1
\end{equation}
Then
\begin{equation}\label{A.9}
  \tfrac{1}{2}(\psi(b)+\psi(\tilde{b})) \ge \tfrac{1}{2}(\psi(a)+\psi(\tilde{a}))\ge \psi(c)
\end{equation}
Moreover, the first inequality is strict unless $\psi'(s)$ is constant on $(\tilde{a},a)$.
\end{lemma}

\begin{proof}  If one takes $a=c$ and then replaces $b$ by $a$, the first inequality becomes the second so it suffices to prove the first one.  Without loss (by translation and scaling) we can take $c=\tfrac{1}{2}, b=1$ so that $\tilde{b}=1$ and $\tilde{a}=1-a$. By the fundamental theorem of calculus (a general convex function is not $C^1$ but it is differentiable with the possible exception of a countable set and the fundamental theorem of calculus holds; see Simon \cite[Theorem 1.28]{SimonConvex})
\begin{equation}\label{A.10}
  \tfrac{1}{2}(\psi(1)+\psi(0))-\tfrac{1}{2}(\psi(a)+\psi(\tilde{a})) = \tfrac{1}{2} \int_{a}^{1} [\psi'(s)-\psi'(1-s)]\, ds
\end{equation}
By convexity, the integrand is non-negative so we have proven \eqref{A.9}.  Moreover if $\psi'(s)$ is not constant on $(1-a,a)$, then the integral is strictly positive.
\end{proof}

\begin{proposition} \lb{PA.6} Let $\psi$, $\overline{\psi}$ and $N$ be as in Theorem \ref{TA.1}.  Then
\begin{equation}\label{A.1A}
  n \equiv \#\{j\,\mid\,\psi\left(\tfrac{j}{N}\right) \le \overline{\psi}\} \ge (N+1)/2
\end{equation}
and
\begin{equation}\label{A.1B}
  \psi\left(\tfrac{1}{2}\right) \le \overline{\psi} \le \tfrac{1}{2}(\psi(0)+\psi(1))
\end{equation}
Moreover, the inequalities in \eqref{A.1B} are strict if $N\ge 2$ and $\psi$ is not an affine function on $[0,1]$ (i.e. $\psi'$ is not constant).
\end{proposition}

\begin{proof} For any $j=0,1,\dots,N$, \eqref{A.9} implies that
\begin{equation}\label{A.11}
   \psi\left(\tfrac{1}{2}\right) \le  \tfrac{1}{2}\left(\psi\left(\tfrac{j}{N}\right)+\psi\left(1-\tfrac{j}{N}\right)\right) \le \tfrac{1}{2}(\psi(0)+\psi(1))
\end{equation}
Averaging over $j$ yields \eqref{A.1B}.  If $\psi$ is not affine on $[0,1]$, then the second inequality is strict for $1\le j \le N-1$ so the second inequality in \eqref{A.1B} is strict.  Since $\psi\left(\tfrac{1}{2}\right) < \tfrac{1}{2}(\psi(0)+\psi(1))$ if $\psi$ is not affine, we see that in the case the first inequality is always strict.

Since $\psi$ is strictly monotone, the first inequality in \eqref{A.1B} implies the unique $x\in [0,1]$ with $\psi(x)=\overline{\psi}$  has $x\ge\tfrac{1}{2}$. This implies that $n = \#\{j\,\mid\,\tfrac{j}{N} \le x\} \ge \#\{j\,\mid\,\tfrac{j}{N} \le \tfrac{1}{2}\} \ge (N+1)/2.$
\end{proof}

\begin{proof} [Proof of Theorem \ref{TA.1}] Let $q=N+1-n\le n$ by \eqref{A.1A}.  Define
\begin{equation}\label{A.13}
  y_j= \overline{\psi} - \psi\left(\tfrac{j}{N}\right) \qquad j=0,\dots,n-1
\end{equation}
\begin{equation}\label{A.14}
  x_j = \left\{
          \begin{array}{ll}
            \psi\left(\tfrac{N+1-j}{N}\right) - \overline{\psi}\, & \hbox{ if } j=1,\dots,q \\
            0, & \hbox{ if } j \ge q
          \end{array}
        \right.
\end{equation}
Since $\psi$ is monotone and $n$ is defined by \eqref{A.1A}, we have that $\mathbf{x}, \mathbf{y}\in\bbR^n_{+,\ge}$.  By the definition of $\overline{\psi}$, we have that
\begin{equation}\label{A.15}
  \sum_{j=1}^{n} x_j = \sum_{j=1}^{n} y_j
\end{equation}

If $N=1$ or $\psi$ is affine on $[0,1]$, it is easy to see that $x_j=y_j$ for all $j$, so, since $\Phi$ is odd, we have that \eqref{A.4} holds.  Thus henceforth we will suppose that  $N\ge 2$ and $\psi$ is not an affine function on $[0,1]$, so, in particular, the inequalities in \eqref{A.1B} are strict.

Note next that because $\psi$ is assumed convex, we have that
\begin{equation}\label{A.16}
  m<p \Rightarrow \psi\left(\tfrac{m+1}{N}\right)-\psi\left(\tfrac{m}{N}\right) \le \psi\left(\tfrac{p+1}{N}\right)-\psi\left(\tfrac{p}{N}\right)
\end{equation}

By the strict form of \eqref{A.1B}, $x_1>y_1$.  Because of \eqref{A.15}, there must be a first $\ell$ so that $x_\ell\le y_\ell$.  We claim that if $\ell<n$, then $x_{\ell+1}\le y_{\ell+1}$.  If $\ell+1>q$, then $x_{\ell+1}=0$ and the required inequality is immediate.  If $\ell+1\le q$, then \eqref{A.16} implies that $x_{\ell}-x_{\ell+1} \ge y_\ell-y_{\ell+1}$.  Subtracting this from  $x_\ell\le y_\ell$ proves that $x_{\ell+1}\le y_{\ell+1}$.  Repeating this argument, proves that for all $j\ge\ell$ we have that $x_j\le y_j$.  Thus by Proposition \ref{PA.4}, $\mathbf{x}\succ \mathbf{y}$.

By Karamata's inequality, \eqref{A.6}, we conclude that $\sum_{j=1}^{n} \Phi(x_j)-\Phi(y_j) \ge 0$.  Since $\Phi$ is odd and $\Phi(0)=0$, this is equivalent to \eqref{A.4}.
\end{proof}

\begin{example} \lb{EA.7} To understand why we need the extra condition \eqref{A.3B} in Theorem \ref{TA.2}, we consider $\mu_S$ for $S=6$, i.e. $13$ pure points with weight $1/13$ at $0,\pm 1,\pm 2,\pm 3,\pm 4,\pm 5,\pm 6$.  By \eqref{5.2}, the average of the square is $A_6 = 14$.  The values of $j^2$ are $j^2=0,1,1,4,4,9,9,16,16,25,25,36,36$ so $n=7$ values are less than $A_6$ and one sees that (ignore $\mathbf{w}$ for now)
\begin{align}\label{A.22}
  \mathbf{x} &= 22,22,11,11,\,\;2,2,0 \nonumber \\
  \mathbf{y} &= 14,13,13,10,10,5,5 \nonumber \\
  \mathbf{w} &= 22,22,\,\;0,11,11,2,2
\end{align}
One can verify that $\mathbf{x}\succ \mathbf{y}$ by hand (and, below, we will prove the result for all $S\ge 2$) but one can't use Proposition \ref{PA.4} as we did in our proof of Theorem \ref{TA.1} for $x_j-y_j$ shifts signs three times instead of one time.  The problem is that the components of $\mathbf{x}$ and $\mathbf{y}$ are paired but shifted.

Look at $\mathbf{w}$ which we get by moving the $0$ from position $7$ to position $3$.  One can handle the first  three partial sums by seeing that $22+22\ge 14+13+13$ and the remaining partial sums by noting that there is only one sign shift after the third place and use Proposition \ref{PA.4} to prove the partial sums of $\mathbf{w}$ dominate those of $\mathbf{y}$ and note it is trivial that partial sums of $\mathbf{x}$ dominate those of $\mathbf{w}$.  The key is that by moving the $0$, the pairs are no longer shifted.
\end{example}

\begin{proof} [Proof of Theorem \ref{TA.2}] Notice that $2N+1$ is odd, so either $\mathbf{x}$ or $\mathbf{y}$ needs to have a zero added.  To be sure that it is added to the $\mathbf{x}$, we need the analog of \eqref{A.1A}, namely that
\begin{equation}\label{A.23}
  n \equiv \#\{j\,\mid\,\psi\left(\tfrac{j}{N}\right) \le \overline{\psi}\} \ge N+1
\end{equation}
If $N=2k$ is even, the $N+1$ values of $j$ among the $2N+1$ overall values with the smallest values of $\psi\left(\tfrac{j}{N}\right)$ are $0,\pm 1,\pm k$, so \eqref{A.23} follows from the first inequality in \eqref{A.1B}.  If $N=2k+1$ is odd, we can't take the max $j$ to be $k$ since that only yields $2k+1<N+1$ values and therefore we need to go up to $j=k+1$ yielding $N+2$ values (leaving over $N-1$), so we need \eqref{A.3C} to hold.  In that case, there are at least three more values $\le\overline{\psi}$ than $>\overline{\psi}$!  (This has to be because $n$ is odd and $N-n$ even and we require that $n>N-n$.)

Define $\mathbf{x}$ to be the $2N+1-n$ values of $\psi-\overline{\psi}$ larger than $0$ written in decreasing order plus $2n-2N+1$ zero values at the end and $\mathbf{y}$ the $n$ non-negative values of $\overline{\psi}-\psi$ so $\mathbf{x}, \mathbf{y}\in\bbR^n$.  As in the proof of Theorem \ref{TA.1}, if we prove that $\mathbf{x}\succ \mathbf{y}$, then \eqref{A.3} follows from Karamata's inequality.

For $N\ge 2$ (so $n\ge 3$), define $\mathbf{w}\in\bbR^n_{+,\ge}$ by
\begin{equation}\label{A.24}
  w_j = \left\{
          \begin{array}{ll}
            x_j, & \hbox{ if } j=1,2 \\
            0, & \hbox{ if } j=3\\
            x_{j-1}, & \hbox{ if } j\ge 3
          \end{array}
        \right.
\end{equation}
We claim first that
\begin{equation}\label{A.25}
  S_n(\mathbf{x}) = S_n(\mathbf{w}); \qquad S_j(\mathbf{x}) \ge S_j(\mathbf{w}),\, j=1,\dots, n-1
\end{equation}
(we do not write this as $\mathbf{x}\succ\mathbf{w}$, first because $\mathbf{w}\notin\bbR^n_{+,\ge}$ and we've only defined majorization for such sequences and also because it is usual to extend $\succ$ to $\bbR^n_+$ by demanding the relations for the decreasing rearrangements rather than by \eqref{A.25}).  This follows because $\mathbf{w}$ is a rearrangement of $\mathbf{x}$ and the decreasing sequence, $S_j(\mathbf{x})$, is maximal among those sums for any arrangement or also, more directly, by noting that $S_j(\mathbf{w}) = S_j(\mathbf{x})$ for $j=1,2,n$ and $S_j(\mathbf{w}) = S_{j-1}(\mathbf{x})$ for $j\ge 3$.

Thus, if we prove that
\begin{equation}\label{A.26}
  S_n(\mathbf{w}) = S_n(\mathbf{y}); \qquad S_j(\mathbf{w}) \ge S_j(\mathbf{y}),\, j=1,\dots, n-1
\end{equation}
then $\mathbf{x}\succ \mathbf{y}$ and the proof is done.

\eqref{A.26} for $j=n$ follows from the definition of $\overline{\psi}$.  For $j=1,2$, we note that by the second inequality of \eqref{A.1B}, we have that $w_1=x_1\ge y_1$ and then that $S_2(\mathbf{w})=S_2(\mathbf{x})=2x_1\ge 2y_1 \ge y_1+y_2$.  For $j=3$, the required inequality is the hypothesis \eqref{A.3B}.

Because of the equality for $j=n$ and inequality for $j=3$, there must be a first $\ell\ge 4$ with $x_\ell \le y_\ell$.  We claim either $x_{\ell+1}=0$ or $x_\ell - x_{\ell+1} \ge y_\ell - y_{\ell+1}$ because if $x_{\ell+1} >0$ and $\ell$ is odd, we have that $x_\ell = x_{\ell+1}$ and $y_\ell = y_{\ell+1}$ or $\ell$ is even and we can use \eqref{A.16} to prove the inequality as in the proof of Theorem \ref{TA.1}.  By induction we see that for $j\ge \ell$, we have that $x_j\le y_j$. By following the argument used to prove Proposition \ref{PA.4} we conclude that for all $j\ge 4$, we have \eqref{A.26}.
\end{proof}

\begin{proof} [Proof of Theorem \ref{T5.2}] If $S$ is half an odd integer, we can use Theorem \ref{TA.1} with $N=S-\tfrac{1}{2}$ and $\psi(x) = \left(\tfrac{1}{2}+Nx\right)^2$ which is non-negative, strictly monotone and convex and $\Psi(y)=y^{2m+1}$.

If $S$ is an integer, we will use Theorem \ref{TA.2} with $N=S$, $\psi(x)=(Sx)^2$ and $\Psi(y)=y^{2m+1}$.  We need to check \eqref{A.3B}, which says that
\begin{equation}\label{A.27}
  2S^2+2 \ge 5A_S = \tfrac{5}{3}S(S+1)
\end{equation}
and \eqref{A.3C} which says that for $S\ge 3$ odd, one has that
\begin{equation}\label{A.28}
  S^2\left(\tfrac{1}{2}+\tfrac{1}{2S}\right)^2 \le A_S = \tfrac{S(S+1)}{3}
\end{equation}

\eqref{A.27} is equivalent to
\begin{equation}\label{A.29}
  0\le S^2-5S+6 = (S-2)(S-3)
\end{equation}
which holds for all integral $S$ (since the polynomial $x^2-5x+6$ is negative precisely for $x\in (2,3)$).  \eqref{A.28} is equivalent to $3S^2+6S+3\le 4S^2+4S$ or
\begin{equation}\label{A.30}
  0\le S^2-2S-3 = (S-3)(S+1)
\end{equation}
which holds for all $S\ge 3$ as required.
\end{proof}

We note that when $S=1$ with the normalization used in this proof, one has that the values of $j^2$ are $0, 1, 1$ so $A_1=\tfrac{1}{3}$ and $\mathbf{x}=(\tfrac{1}{3},\tfrac{1}{3})$ while $y=(\tfrac{2}{3},0)$ so $\mathbf{x}\prec \mathbf{y}$ and for all $m$, \eqref{5.6} holds with the opposite sign!

\smallskip

As we mentioned, we would guess that results like Theorem \ref{TA.2} hold without \eqref{A.3B}.  To be explicit, consider the an analog of \eqref{5.6} with $3j^2$ replace by $|j|^p$ and $3A_S=S(S+1)$ replaced by the suitable average.  We know this analog is valid for any $p>1$ if $S$ is half an odd integer but because \eqref{A.3B} fails when $p<3/2$ and $S$ is large we don't know the answer to

\textbf{Question 2} Does the analog of \eqref{5.6} hold for $S\ne 1$ integral when $3j^2$ is replaced by $|j|^p$ for any $p>1$ and with $3A_S=S(S+1)$ replaced by the suitable average.

This is an explicit example for what we hope might be a general result.

%%%%%%%%%%%%%%%%%%%%%%%%%%%%%%%

\end{document}